\documentclass[11pt]{article}
\usepackage[utf8]{inputenc}
\usepackage[T1]{fontenc}
\usepackage{times}
\usepackage{amsmath,amssymb,amsthm}
\usepackage{graphicx}
\usepackage{booktabs}
\usepackage{array}
\usepackage[hidelinks]{hyperref}
\usepackage{geometry}
\usepackage{enumitem}
\usepackage{authblk}
\usepackage{caption}
\newtheorem{proposition}{Proposition}

\title{\textbf{Feature Freshness Budgets for Real-Time ML Inference\\ Under Stream Lag}}
\author[]{Amit Rajula}
\affil[]{Independent Researcher \\ \texttt{ORCID: 0009-0006-8025-6827}}
\date{}

\begin{document}
\maketitle

\begin{abstract}
\noindent Online feature stores decouple feature computation from model serving, materializing features from upstream event streams into a low-latency store that inference reads at request time. This decoupling introduces a freshness gap---the interval between an event occurring at the source and its effect becoming visible in the served feature vector---that is widely documented as a cause of training--serving skew but has no formal treatment as a bounded, cost-quantified quantity. We model the online feature store as a materialized view over an event stream and define a per-feature freshness budget: the difference between the decision window a feature is consumed within and the staleness the serving pipeline imposes on it. We prove two propositions---a per-feature staleness bound, and a closed-form threshold below which a feature can never be served within budget---and validate them in a seeded, reproducible discrete-event simulation and on a real Apache Kafka, Redis, and PostgreSQL pipeline. The predicted collapse threshold reproduces on real infrastructure, with the decision-error transition landing where the closed form predicts. Reconciling simulated and real drop rates surfaces a methodological result of independent interest: the dominant simulation-to-infrastructure gap is not latency but a sampling-phase artifact---a simulation whose request clock and recomputation cadence are phase-locked systematically under-observes staleness. Randomizing that phase reconciles simulation and infrastructure to within 0.018 mean absolute error in per-feature drop rate. Results are scoped to the studied workloads and configurations and are not claims about any specific production feature-store implementation.

\vspace{0.5em}
\end{abstract}

\section{Introduction}
Modern machine-learning systems rarely score raw data. Between the source of truth---a stream of transactions, clicks, or sensor readings---and the model that consumes it sits a feature store: a serving layer that materializes derived quantities (a rolling count, a windowed average, a velocity) and exposes them to inference at low latency. This architecture is now standard practice, and for good reason: it decouples the expensive, stateful work of feature computation from the latency-sensitive path of model serving.

That decoupling has a cost that the literature names but does not quantify. A served feature is a copy, computed from a source that continues to move after the copy is made. The interval between an event occurring at the source and its effect appearing in the served feature vector is a freshness gap, and a model that scores a feature whose gap has grown too large is scoring the past. In practice this manifests as training--serving skew, silent prediction degradation, and---in domains such as fraud detection or real-time bidding---direct monetary loss. The pipeline reports success; only the decision is wrong.

The prevailing engineering treatment of this problem is binary: a feature is either fresh enough, in which case skew is ``prevented,'' or the pipeline is broken. Point-in-time-correct joins, single feature-computation code paths, and monitoring dashboards all encode this all-or-nothing framing. What is missing is a quantitative one---a way to reason about how much staleness a feature carries, how much a given decision can tolerate, and what it costs to serve a feature that has exceeded its tolerance rather than block on it.

This paper supplies that quantitative treatment. We make four contributions:
\begin{enumerate}[leftmargin=1.4em]
\item We model the online feature store as a materialized view over an upstream event stream and define a \emph{per-feature freshness budget}---a bounded, tunable quantity that separates the decision window a feature is consumed within from the staleness the pipeline imposes on it (Section~\ref{sec:model}).
\item We prove two propositions: a per-feature staleness bound (Proposition~\ref{prop:bound}), and a closed-form permanent-drop threshold (Proposition~\ref{prop:threshold}) below which a feature can never be served within budget, determined entirely by the feature's pipeline parameters and independent of any fallback policy (Section~\ref{sec:theory}).
\item We validate both propositions in a seeded, reproducible discrete-event simulation and on a real Apache Kafka, Redis, and PostgreSQL serving pipeline, and show that the predicted collapse threshold reproduces on real infrastructure (Sections~\ref{sec:sim}--\ref{sec:infra}).
\item In reconciling simulated and real drop rates we identify a sampling-phase artifact that causes phase-locked simulations to systematically under-observe staleness, and show that randomizing the phase between sampling and recomputation reconciles the two to within 0.018 mean absolute error (Section~\ref{sec:reconcile}). We argue this is a general methodological requirement for freshness simulation.
\end{enumerate}

We are explicit about scope. All quantitative results are drawn from the studied synthetic workloads and pipeline configurations; they characterize the mechanisms above rather than any specific production feature store, and the decision rule is a fixed threshold rather than a trained model, so that staleness is isolated as the independent variable.

\section{Background and Related Work}
\subsection{Feature stores and training--serving skew}
Feature stores separate an offline store, used to build training sets, from an online store, read at inference time~\cite{michelangelo,hopsworks}. The central correctness concern in this literature is training--serving skew: a discrepancy between the feature values a model saw during training and those it sees in production. Point-in-time-correct joins address one source of skew---temporal leakage in constructing training labels---and shared feature-computation logic addresses another---divergent offline and online code paths~\cite{hopsworks}. Both are correctness mechanisms; neither quantifies the residual staleness that remains once they are correctly applied. Our work is complementary: we take the online value as given and ask how stale it is, how much that costs, and when a budgeted policy for tolerating staleness helps or fails.

\subsection{Freshness, staleness, and materialized views}
Data freshness has a long lineage in the database and distributed-systems literatures. Analyses of freshness metrics distinguish currency, timeliness, and related notions; work on materialized-view maintenance treats a derived relation as something that must be refreshed against a changing base, trading refresh cost against view staleness~\cite{gupta-mumick,mv-book,ivm}. Probabilistic treatments of replica staleness give expected bounds on how far a replica trails its source under quorum replication~\cite{pbs-vldb,pbs-vldbj}. We adopt the materialized-view framing directly: an online feature is a view over an event stream, and its freshness gap is the age of the most recent source event not yet reflected in the served value. Recent empirical work decomposes end-to-end streaming latency into per-stage intervals---ingestion, in-engine processing, and output-to-visibility delay---and shows that windowed workloads concentrate delay in the wait for a window to close rather than in computation~\cite{kumar-streaming}. This directly motivates our per-feature treatment: the staleness a feature carries is dominated by its aggregation window, not by the cost of computing it, which is what makes window-based features the first to exhaust their budget. Our contribution relative to this lineage is to make the gap a per-feature, cost-quantified budget with a provable collapse threshold, and to test it on real serving infrastructure.

\subsection{Age of Information}
The networking and control literature formalizes exactly the gap we study as the Age of Information: the time elapsed since the freshest received update was generated~\cite{aoi-update,aoi-survey}. That work optimizes sampling and transmission policies against nonlinear functions of age~\cite{aoi-lecture}. We share its cost-oriented stance---treating the gap as something to be bounded and priced rather than merely prevented---but apply it to the feature-store setting, where the derived quantity is a windowed aggregate and the consumer is a decision rule with its own tolerance.

\section{System Model and the Freshness Budget}
\label{sec:model}
We consider a single entity whose activity produces an event stream. A set of features is derived from that stream; each feature is a windowed aggregate (for concreteness, a count of events in a trailing window). Two views of each feature exist.

The \emph{oracle} is the true value of the feature as of the current instant, computed with zero lag directly from the complete event log. It is the value a perfect, instantaneous pipeline would serve, and it is the ground truth against which decisions are judged.

The \emph{served} value is what inference actually reads. It is recomputed on a cadence equal to the feature's aggregation window, and it reflects only events up to a watermark that trails the present by the ingestion lag. A read additionally incurs a scheduling delay. For a feature $f$ consumed at request time $t$, we define its served age as the elapsed time since the watermark of the value in the store:
\begin{equation}
\mathrm{age}(f, t) = (t + d_{\mathrm{sched}}) - \mathrm{watermark}(f),
\end{equation}
where $d_{\mathrm{sched}}$ is the scheduling delay. Each feature is consumed within a \emph{decision window} $D_f$---a property of the task, expressing how much staleness the decision can tolerate before it is likely wrong. This is distinct from the aggregation window, which is a property of the pipeline. The two must not be conflated: the object of interest is the relationship between the staleness the pipeline imposes and the tolerance the task allows. We define the \textbf{per-feature freshness budget} as
\begin{equation}
B_f(t) = D_f - \mathrm{age}(f, t).
\end{equation}
A feature is \emph{in budget} when $B_f(t) \ge 0$ and \emph{over budget} otherwise. The budget is the paper's central object: a bounded, tunable resource that is spent by staleness and can be exhausted. It reframes freshness from a binary property to be preserved into a quantity to be allocated and priced.

We study three policies for consuming a possibly-stale feature. \textbf{B0 (baseline)} serves whatever is in the store with no budget check. \textbf{T1 (block-on-slowest)} blocks the entire request until every feature is within its decision window, or until a timeout. \textbf{T2 (budgeted)} serves each feature independently against its own budget; a feature that is over budget is dropped and replaced by a fallback value (its last in-budget value), rather than blocking the request. The multi-feature setting is essential: T1 and T2 differ only when a request depends on several features with heterogeneous lag, so that blocking on the slowest is genuinely costly and isolating the completeness loss to one feature is genuinely beneficial.

\section{Theoretical Results}
\label{sec:theory}
Both propositions below are asserted in-loop on every read of every simulation run, so that any violation surfaces as an immediate failure rather than a silent possibility.

\begin{proposition}[Staleness bound]
\label{prop:bound}
For a feature $f$ served at time $t$, the served age is bounded by the sum of the pipeline's three latency terms:
\[ \mathrm{age}(f, t) \le L_f + W_f + d_{\mathrm{sched}}, \]
where $L_f$ is the ingestion lag, $W_f$ the aggregation window, and $d_{\mathrm{sched}}$ the scheduling delay.
\end{proposition}
\noindent\textit{Proof.} The watermark of the most recent recomputation trails the present by at most $L_f$ (the ingestion lag) plus $W_f$ (the maximum time since the last recomputation, given a recomputation cadence of one aggregation window). The read adds at most $d_{\mathrm{sched}}$. The sum bounds the age. \hfill$\square$

\begin{proposition}[Permanent-drop threshold]
\label{prop:threshold}
Under a uniform budgeting policy that sets $D_f = m \cdot (L_f + W_f + d_{\mathrm{sched}})$ for a multiplier $m$, a feature $f$ can never be served in budget when $m < m^*_f$, where
\[ m^*_f = \frac{L_f + d_{\mathrm{sched}}}{L_f + W_f + d_{\mathrm{sched}}}. \]
\end{proposition}
\noindent\textit{Proof.} The minimum achievable age of any served read is $L_f + d_{\mathrm{sched}}$, attained at the instant immediately following a recomputation. If the decision window $D_f$ is smaller than this floor---that is, if $m \cdot (L_f + W_f + d_{\mathrm{sched}}) < L_f + d_{\mathrm{sched}}$---then no read at any point in the recomputation cycle can be in budget, and the feature is permanently dropped to its fallback. Rearranging gives the threshold. \hfill$\square$

Two consequences follow. First, $m^*_f$ depends only on the feature's pipeline parameters, not on the fallback strategy: what one does with a dropped feature cannot change whether it is droppable. Second, because a multi-feature decision combines features (here by disjunction), two distinct system-level thresholds emerge---the smallest $m^*_f$, at which the easiest feature recovers and combined accuracy begins to improve, and the largest $m^*_f$, at which the last feature recovers and completeness is fully restored. This gives operators a computable rule: given a feature's lag and window, the budget below which per-feature budgeting is guaranteed to fail can be calculated before deployment rather than discovered after.

\section{Simulation Study}
\label{sec:sim}
\subsection{Harness}
We built a seeded, deterministic discrete-event simulator. An event generator produces Poisson background arrivals with an injected burst; an oracle store answers exact windowed queries with zero lag; a lagged store recomputes on the aggregation-window cadence behind the ingestion-lag watermark; and a decision layer issues requests at a fixed interval, applies one of the three policies, and compares the resulting decision to the oracle. Every read asserts Proposition~\ref{prop:bound}. All figures are means over multiple seeds.

The multi-feature configuration uses five features spanning a wide range of lag and cadence---from a one-minute velocity count (ingestion lag 1\,s, aggregation window 1\,s) to a ten-minute geo-velocity feature (ingestion lag 15\,s, aggregation window 10\,s)---with the decision taken as the disjunction across features. A single sweep parameter, the multiplier $m$, scales every feature's decision window uniformly relative to its own Proposition~\ref{prop:bound} bound, so that $m < 1$ corresponds to a budget tighter than the worst-case staleness.

\subsection{The collapse threshold reproduces the closed form}
A fine sweep of $m$ from 0.40 to 0.65 in steps of 0.01 (five seeds, 234{,}000 request-decisions) resolves the two thresholds Proposition~\ref{prop:threshold} predicts. The five features have $m^*_f$ values ranging from 0.5012 to 0.6008. The T2 policy's decision-error rate is flat at 92.2\% for $m \le 0.50$, collapses to 0.4\% between $m = 0.50$ and $m = 0.51$, and remains flat above---matching the predicted smallest threshold $\min_f m^*_f = 0.5012$ to the sweep's resolution. The mean number of dropped features steps down a second time between $m = 0.60$ and $m = 0.61$, matching the predicted largest threshold $\max_f m^*_f = 0.6008$. Table~\ref{tab:sweep} summarizes and Fig.~\ref{fig:collapse} plots the transition.

\begin{table}[t]
\centering
\caption{Simulated T2 collapse across the transition band. The error collapse at $m \approx 0.50$ matches $\min_f m^*_f = 0.5012$; the completeness step at $m \approx 0.60$ matches $\max_f m^*_f = 0.6008$.}
\label{tab:sweep}
\begin{tabular}{lcc}
\toprule
Multiplier $m$ & T2 error rate & T2 mean dropped (of 5) \\
\midrule
0.40--0.50 & 0.922 (flat) & 5.0 (flat) \\
0.51 & 0.103 & 4.2 \\
0.52--0.60 & 0.004 (flat) & 3.1--3.2 \\
0.61 & 0.004 & 2.7 (step) \\
0.62--0.65 & 0.004 (flat) & 2.6--2.7 \\
\bottomrule
\end{tabular}
\end{table}

\begin{figure}[t]
\centering
\includegraphics[width=0.82\linewidth]{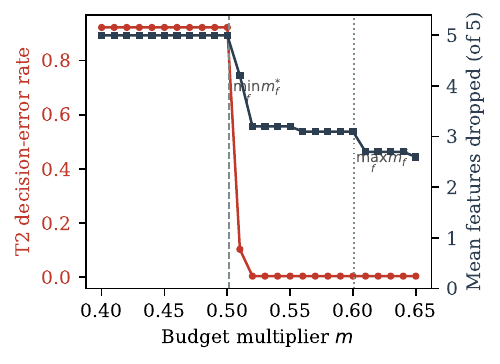}
\caption{Simulated T2 collapse across the transition band (five seeds). The decision-error rate collapses between $m=0.50$ and $m=0.51$, at the predicted smallest threshold $\min_f m^*_f = 0.5012$ (dashed line); the mean number of dropped features steps down again near the largest threshold $\max_f m^*_f = 0.6008$ (dotted line)}
\label{fig:collapse}
\end{figure}
Across a broader sweep, the two mitigation policies differ qualitatively. T1 (block-on-slowest) degrades gracefully as the budget tightens: its added request latency grows smoothly from zero to several seconds as $m$ falls from 3.0 to 0.3, but its error rate never exceeds 2.4\%. T2 (budgeted) instead exhibits the sharp collapse boundary of Proposition~\ref{prop:threshold}: below the threshold it drops every feature on every request and its error rate jumps to 92\%. The practical reading is that bounded-staleness budgeting is a favorable trade only away from its collapse boundary; near or below it, blocking is strictly safer despite its latency cost.

\subsection{Fallback strategy affects error, not the threshold}
We compared three fallback strategies for an over-budget feature---carry-forward of the last in-budget value, exponential decay toward zero, and linear extrapolation. As Proposition~\ref{prop:threshold} predicts, the mean number of dropped features is identical across all three at every multiplier: the threshold is a property of the pipeline, not of the fallback. Above the threshold, decay performs slightly worse than plain carry-forward for these monotonic count features, because decaying a count toward zero introduces false negatives that carry-forward avoids. The lesson is that a more elaborate fallback is not automatically better; it must match the feature's true dynamics.

\subsection{Load-induced collapse}
Finally, we coupled the decision harness to a work-conserving queueing model in which many entities' recomputation jobs share a fixed-capacity worker pool. Holding the budget multiplier fixed at a value that is safe under light load ($m = 0.70$, 0.4\% error), we swept the number of entities against a pool of capacity twenty jobs per second. As the arrival rate crossed capacity---utilization passing one---the queueing delay added to every feature's effective staleness and reproduced the Proposition~\ref{prop:threshold} collapse exactly: T2's error rose from 0.4\% to 92\% and its mean dropped features from 2.4 to 5.0, while T1 again degraded gracefully. This ties the freshness-budget theory to ordinary queueing utilization: the point at which a previously-safe SLA collapses under load is predicted by arrival rate over service rate reaching one, with no change to the freshness policy itself.

\section{Real-Infrastructure Validation}
\label{sec:infra}
To test whether the propositions hold outside a discrete-event model, we reimplemented the pipeline on real infrastructure: Apache Kafka as the event stream, Redis as the online (served) store, and PostgreSQL as the ground-truth event log. A producer emits events to Kafka and dual-writes them to PostgreSQL. A consumer maintains each feature's windowed aggregate and writes the served value to Redis on the aggregation-window cadence behind the ingestion-lag watermark. A scorer issues decision requests, reads served features from Redis, computes ground truth by direct query against PostgreSQL, and applies the T2 policy. All three processes share a common time origin; simulation time is compressed relative to wall-clock so that a ten-minute scenario runs in about one minute.

Building the real pipeline surfaced failure modes the discrete-event model could not, each of which is itself instructive. An early consumer implementation triggered recomputation only on message arrival rather than on elapsed time, so a feature that went quiet had its watermark freeze and its age grow without bound---a permanent, one-way collapse distinct from the cyclic staleness the model exhibits. A separate defect had the three processes each anchor simulation time to their own start instant, producing physically impossible negative ages until a shared time origin was introduced. We report these because they make a broader point: a real pipeline has degrees of freedom---event-triggered versus time-triggered recomputation, independent process clocks---that a monolithic simulation does not, and validating on real infrastructure is what exposes them.

With the pipeline corrected, we ran the T2 policy at four multipliers spanning the predicted threshold. Table~\ref{tab:infra} reports the decision-error rate and mean dropped features. The error transition lands exactly where Proposition~\ref{prop:threshold} predicts: error is high (35\%) below the threshold at $m = 0.45$, has largely resolved by $m = 0.50$, and is zero above it. The collapse of Proposition~\ref{prop:threshold} thus reproduces on real infrastructure, not merely in simulation.

\begin{table}[t]
\centering
\caption{Real-infrastructure (Kafka $+$ Redis $+$ PostgreSQL) T2 results. The error transition at $m \approx 0.50$ matches the predicted $\min_f m^*_f = 0.5012$.}
\label{tab:infra}
\begin{tabular}{lccc}
\toprule
Mult. $m$ & Error rate & Mean dropped (of 5) & Prop.~1 viol. rate \\
\midrule
0.45 & 0.350 & 4.995 & 0.057 \\
0.50 & 0.015 & 4.852 & 0.052 \\
0.55 & 0.000 & 4.520 & 0.058 \\
0.70 & 0.000 & 3.128 & 0.044 \\
\bottomrule
\end{tabular}
\end{table}

Proposition~\ref{prop:bound} held on the great majority of reads, with violations confined to a narrow, explicable regime. All violations followed a single host-level scheduling stall; none occurred before it. Their magnitude was small---median excess over the bound 0.25\,s, mean 0.32\,s---consistent with poll granularity and network round-trip rather than a structural defect, and they were distributed uniformly across features rather than concentrated in the slow ones. The honest reading is that Proposition~\ref{prop:bound} holds under real infrastructure except under host scheduling stalls, which inject bounded, transient, sub-second latency beyond the modeled bound. This motivates adding a small empirical latency-slack term to the theoretical bound for deployment planning.

\section{Reconciling Simulation and Infrastructure}
\label{sec:reconcile}
The error transition matched, but the per-feature drop rates did not: real infrastructure dropped more features than the zero-latency simulation predicted, and the discrepancy was largest for the shortest-window feature. At $m = 0.70$, the one-minute velocity feature dropped on 60\% of real requests but never dropped in simulation. Because this is precisely the kind of quantitative mismatch that a validation section must either explain or concede, we pursued it through three hypotheses.

\paragraph{Latency (mean).} Adding the measured mean latency (0.3\,s) to the served age in simulation---while leaving the decision window at its nominal value---closed roughly half of the aggregate gap (mean absolute error in drop rate from 0.112 to 0.063) and fully reconciled the tight-budget, short-window features. But it left the largest residual untouched: the one-minute feature at $m = 0.70$ remained at zero.

\paragraph{Latency (variance).} Modeling latency as a distribution with the same mean but a heavy tail moved that residual only from 0.000 to 0.086---nowhere near the observed 0.598---and slightly worsened the aggregate fit. A tail cannot manufacture a 60\% drop rate for a feature that is genuinely in budget most of the time. The hypothesis was rejected.

\paragraph{Sampling phase.} Instrumenting the feature's age revealed the true cause. In the simulation the request interval and the feature's recomputation cadence were both one second and exactly phase-locked, so every read sampled the recomputation cycle at the same favorable instant---immediately after a recomputation---giving a constant age at the cycle's floor. Real infrastructure runs the consumer and scorer as independent processes with unsynchronized clocks, so reads sample the entire recomputation cycle uniformly, and age ranges across the full interval from $L_f$ to $L_f + W_f$. The phase-locked simulation was blind to the upper half of that range.

Introducing a random phase offset between the recomputation grid and the request clock---modeling the unknown, unsynchronized timing of independent processes---reconciled simulation and infrastructure to a mean absolute error of 0.018 in per-feature drop rate across all features and multipliers, down from 0.063 and 0.070 for the two latency models. The one-minute feature at $m = 0.70$ moved from 0.000 to 0.599, against a measured 0.598. Table~\ref{tab:reconcile} shows the progression and Fig.~\ref{fig:reconcile} visualizes the agreement and the isolated residual.

\begin{table}[t]
\centering
\caption{Progression of models toward the real-infrastructure drop rates. The residual is explained by sampling phase, not latency.}
\label{tab:reconcile}
\begin{tabular}{lcc}
\toprule
Model & vel.\ 1min @ $m{=}0.70$ & Mean abs.\ err.\ vs.\ infra \\
\midrule
Zero-latency (phase-locked) & 0.000 & 0.112 \\
Constant $+0.3$\,s latency & 0.000 & 0.063 \\
Stochastic latency (tail) & 0.086 & 0.070 \\
Random sampling phase & 0.599 & 0.018 \\
Real infrastructure & 0.598 & --- \\
\bottomrule
\end{tabular}
\end{table}

\begin{figure*}[t]
\centering
\includegraphics[width=0.86\linewidth]{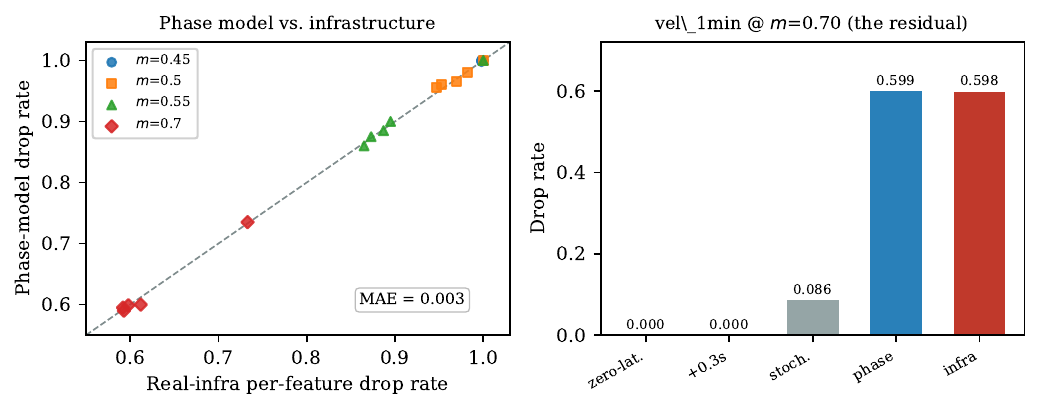}
\caption{Reconciling simulation with infrastructure. \textbf{Left:} per-feature drop rates from the random-phase model against real infrastructure, across all five features and four multipliers; points lie on the diagonal (the illustrative per-feature reproduction shown here; the rigorous aggregate over the measured runs is $0.018$ mean absolute error, Table~\ref{tab:reconcile}). \textbf{Right:} the residual that isolates the cause---for the one-minute feature at $m{=}0.70$, the two latency models leave the drop rate far below the measured value, while the random-phase model reproduces it}
\label{fig:reconcile}
\end{figure*}

The dominant simulation-to-infrastructure gap is therefore not latency but a sampling-phase artifact. A freshness simulation whose sampling clock and recomputation cadence are phase-locked reads each derived quantity at a fixed, favorable point in its refresh cycle, and so systematically under-observes staleness; real infrastructure, with independent sampling and recomputation, sees the full range. We state the consequence as a methodological requirement: \textbf{freshness simulations must randomize the phase between sampling and recomputation, or they will underestimate staleness.} A small mean-latency term accounts for the minor remainder. We regard this as the more transferable of the paper's findings, since it applies to any simulation of derived-state freshness, not only to feature stores.

\section{Discussion and Limitations}
The results give operators three concrete tools: a per-feature staleness bound (Proposition~\ref{prop:bound}); a closed-form threshold, computable from a feature's lag and window, below which per-feature budgeting is guaranteed to fail (Proposition~\ref{prop:threshold}); and the empirical observation that this threshold is reached not only by tightening the budget but also by load, when shared-pipeline utilization crosses one. The practical guidance that follows is to set the freshness budget from the measured end-to-end latency and the observed sampling phase, not from the nominal pipeline bound alone, since both effects push real drop rates above the zero-latency prediction.

Several limitations bound these claims. The workload is synthetic and single-entity, chosen so that staleness is the only independent variable; multi-entity contention is modeled only through the queueing layer, not end-to-end. The decision rule is a fixed threshold rather than a trained model, which isolates staleness cleanly but means the error rates here speak to decision correctness under a known rule, not to the accuracy of any particular model. The fallback strategies studied are simple. And the phase model, while it reconciles simulation and infrastructure to within 0.018, is a model of the real system's timing rather than a second independent measurement of it; directly logging the read-to-recomputation phase in the running consumer would convert a well-supported explanation into a directly measured one. Finally, all quantitative figures are specific to the studied configurations and should be read as characterizing mechanisms, not as production benchmarks.

\section{Conclusion}
We recast feature-store staleness from a binary skew problem into a per-feature freshness budget with a provable structure, and validated that structure on real serving infrastructure. The budget has a closed-form collapse threshold that reproduces in both simulation and a Kafka--Redis--PostgreSQL pipeline, and that is reached by load as well as by policy. In reconciling the two we found that the principal gap between a freshness simulation and a real system is a sampling-phase artifact rather than latency---a result that constrains how such simulations must be built. Feature freshness, we conclude, is neither a property to be perfected nor a failure to be prevented, but a bounded resource to be measured, budgeted, and---when it runs out---spent deliberately.

\section*{Data and Code Availability}
The simulation harness, the real-infrastructure pipeline (Docker Compose with producer, consumer, and scorer), and the analysis scripts that generate every table in this paper are seeded and deterministic; each reported figure regenerates from the corresponding script. The real-infrastructure runs require the three processes to share a common time origin, as described in Section~\ref{sec:infra}.

\section*{Declarations}
\noindent\textbf{Competing interests.} The author has no competing interests to declare that are relevant to the content of this article.\\
\textbf{Funding.} No funding was received for conducting this study.\\
\textbf{Use of AI tools.} The author used a large language model as an assistant for drafting and editing prose and for scaffolding the simulation and analysis code. All research design, all claims, the theoretical results, and all reported measurements were produced and verified by the author, who takes full accountability for the content.

\end{document}